\documentclass[final, 10pt, twocolumn, conference]{IEEEtran}
\usepackage[left=0.625in,right=0.625in,top=0.75in,bottom=1in]{geometry}
\usepackage{setspace}
\usepackage{amsmath}
\usepackage{amsthm}
\usepackage{amssymb}
\usepackage{cite}
\usepackage{graphicx}
\usepackage{epstopdf}
\usepackage{url}
\usepackage{cite}
\usepackage{texdef2015}
\usepackage{color}
\usepackage{caption}
\usepackage{subcaption}
\usepackage{float}
\allowdisplaybreaks
\usepackage{tikz}
\usetikzlibrary{automata,arrows,positioning,calc,fit,shapes.multipart,chains,shapes}
\usepackage{microtype}
 
\newtheorem{theorem}{Theorem}
\newtheorem{corollary}{Corollary}

\newtheorem{lemma}{Lemma}
\newtheorem{definition}{Definition}

\newcommand{\age}{\Delta}

\newcommand{\negfigspace}{\vspace{-3mm}}

\newcommand{\cf}{c}

\newcommand{\eos}[2]{\mu_{#1:#2}}
\newcommand{\varos}[2]{\sigma^2_{#1:#2}}

\newcommand{\agep}{\Delta_P}
\newcommand{\agee}{\Delta_E}

\begin{document}
% \title{The Age of Information vs.\ The Quorum Size}
% \title{Age of Information in Distributed Systems}
\title{Multicast With Prioritized Delivery: \nn
How Fresh is Your Data?}
\author{Jing Zhong, Roy D.~Yates and Emina Soljanin\\
 \small 
Department of ECE, Rutgers University, \{jing.zhong, ryates, emina.soljanin\}@rutgers.edu}

% wipe out the copyright box
\makeatletter
\def\@copyrightspace{\relax}
\makeatother

\maketitle

\begin{abstract}
We consider a multicast network in which real-time status updates generated by a source are replicated and sent to multiple interested receiving nodes through independent links.
The receiving nodes are divided into two groups: one priority group consists of $k$ nodes that require the reception of every update packet, the other non-priority group consists of all other nodes without the delivery requirement.
Using age of information as a freshness metric, we analyze the time-averaged age at both priority  and non-priority nodes.  
For shifted-exponential link delay  distributions, the average age at a priority node is lower than that at a non-priority node due to the delivery guarantee. 
However, this advantage for priority nodes disappears if the link delay is exponential distributed.
Both groups of nodes have the same time-averaged age, which implies that the guaranteed delivery of updates has no effect the time-averaged freshness.
%For non-priority nodes, waiting for every update delivery at the priority nodes doesn't 

\end{abstract}

\section{Introduction} 
The analysis of information freshness arises from a variety of real-time status updating systems, in which update messages generated by the sources are sent to interested receivers through a communication system. 
For instance, the real-time information updates of autonomous cars are broadcast to nearby vehicles and infrastructures.
Similarly, live video captured for remote surgery is required to be available at the doctor with ultra low delay. 
In these systems, the knowledge of the source state at the receiver is desired to be as fresh as possible. 
This leads to the introduction and analysis of an ``Age of Information" (AoI) freshness metric \cite{Kaul2012infocom,Kam2013, Costa2014,Sun2016,Najm2016,Bedewy2016,Kadota2016,Yates2017}.
Age of information, or simply \emph{age}, measures the time difference between now and when the most recent update was generated.
At any time $t$, if the most recent update at the receiver is generated at time $u(t)$, then the instantaneous age at the receiver is $t-u(t)$.

In early work on age analysis \cite{Kaul2012infocom}, it was shown that the source should limit its update rate in order to avoid queueing delay caused by overloading the system with first-come first-served (FCFS) policy.
Given the observation of unnecessary waiting in FCFS systems, subsequent research looked at last-come first-served (LCFS) queueing systems that discard older updates as soon as a new update comes \cite{Costa2014,Najm2016}, and last-generated first-server (LGFS) policy with preemption in service for multihop networks \cite{Bedewy2016}. 
When the source has no knowledge of the service system state, allowing packet preemption at the queue provides lower average age in general.
However, these results are limited to the case where the update arrival process is given.
In \cite{Sun2016}, the authors consider a different scenario where the system state is available at the source such that a new update is generated only after the service of the previous update is completed. 
A lazy updating scheme is proved to be age-optimal, indicating that the source should wait for a short period before sending a new update if the service time of the previous update is too small.
The analysis in \cite{Sun2016} applies only to systems without preemption.

In this work, we consider an update multicast system in which real-time update messages generated by the source are broadcast to a set of nodes through i.i.d. links with random network delays.
The receiving nodes are categorized into two groups. The \emph{priority} group consists of nodes that require the delivery of every update, while all other nodes without the delivery requirement are regarded as the \emph{non-priority} group.
Once a node receives an entire update message, it acknowledges the source by sending  instantaneous feedback.
This model arises in a variety of delay-sensitive applications, e.g. vehicle networks where the update messages are popular and simultaneously request by large numbers of users.
%Although the update message is desired to timely, 
Some receiving nodes require the history of all updates for the purpose of data aggregation and processing; thus the delivery of every update message is crucial.  

Our work is closely related to the LCFS and LGFS systems with preemption and state-dependent updating. 
Since any new update generated by the source leads to the termination of the previous update, each link is equivalent to a LCFS queue with preemption in service. 
Thus, the instantaneous feedback enables the source to submit new updates based on the state of the queue, either replacing staled updates with a fresh update or waiting for the service of the current update to be completed.
We assume each update packet is divided into small chunks and encoded with a rateless code to overcome channel erasure in the multicast network. 
In this case, the number of chunks corresponding to an update is required to  reach a certain minimum level for the update to be successfully decoded.
Moreover, the source is also able to instantaneously terminate an update in the middle of transmission.
Here we consider a simple updating scheme that exploits instantaneous feedback.
Once the current update is delivered to all the nodes in the priority group, the source terminates the transmission of the current update and broadcasts a fresh update. 
Our goal here is to evaluate the average age for both types of nodes.

%In this work, the system model and age metric definition is described in section \ref{sec:system}.
%We derive the average age for a priority node and a non-priority node in section \ref{sec:priority} and \ref{sec:nonp}, respectively. 
%For shifted exponential link delay distribution, we obtain a lower bound on the average age for the priority node using the order statistics of random variables.
%We show that both groups have the same average age if the delay is exponentially distributed.
%Some simulation results are provided in section \ref{sec:evaluation} and we conclude in section \ref{sec:conclusion}.

\section{Problem Formulation}\label{sec:system}
% \textcolor{blue}{We need to say here what we want to find.}
\begin{figure}[t]
\centering\small
\begin{tikzpicture}[node distance=1cm]
\node [draw,circle, rounded corners,align=center, thick] (newsource) {Source};
% \node[draw,rectangle,align=center,minimum height=4mm] (newupdate)[above right = 0.01 and 0.1 of newsource] {j+1};
\node[draw,rectangle,rounded corners,align=center,minimum height=4mm,thick, fill=lightgray] (node_2)[right = 4 of newsource.north] {Node 2};
\node[draw,rectangle,rounded corners,align=center,minimum height=4mm,thick, fill=lightgray] (node_1)[above = 0.5 of node_2.west] {Node 1};
\node[draw,rectangle,rounded corners,align=center,minimum height=4mm,thick, fill=lightgray] (node_3)[below = 1 of node_2.east] {Node $k$};
\node[draw,rectangle,rounded corners,align=center,minimum height=4mm,thick] (node_n)[below = 0.6 of node_3.west] {Node $k$+1};
\draw  [->,thick] (newsource.east) node[draw,rectangle,minimum height=3mm,thin,opacity=1] [above right = 0.5 and 0]{$j$+1} .. controls ++(1,0.5) .. (node_1.west) node[draw,rectangle,minimum height=3mm,thin] [above left = 0.1 and 0.2]{$j$};
%\draw[->,thick] (newsource.east) -- node[draw,rectangle,minimum height=3mm,thin][above left]{$j$+1} ++(1.5,0) |-  (node_1.west) node[draw,rectangle,minimum height=3mm,thin] [above left = 0 and 0.2]{$j$};
\draw[->,thick] (newsource.east) .. controls ++(2,-0.5) .. (node_3.west) node[draw,rectangle,minimum height=3mm,thin,opacity=1] [above left = 0.1 and 0.5]{$j$};
\draw[->,thick] (newsource.east) .. controls ++(1,0.2) .. (node_2.west) node[draw,rectangle,minimum height=3mm,thin,opacity=1] [above left = 0 and 0.6]{$j$}; 
\draw[->,thick] (newsource.east) .. controls ++(1.2,-0.8) .. (node_n.west) node[draw,rectangle,minimum height=2mm,thin] [above left = 0.3 and 0.3]{$j$};
\path (node_3.north) -- node {$\vdots$} ++(-0.5,1);
%\draw[red,dashed] (4.4,1.2) node [right] {$\wqr$ nodes have update $j$} rectangle (3,-2.1) ; 
% \node[draw,circle,rounded corners,align=center,] (client)[right = 2 of node_2] {client};
% \draw[->,thick] (node_1.east) -- (client.west);
% \draw[->,thick] (node_3.east) -- (client.west);
% \draw[->,thick] (node_n.east) -- (client.west);

\end{tikzpicture}
\caption{System diagram: source broadcasts status updates to $n$ nodes through i.i.d. channels. The $k$ nodes in the priority group are shaded. The transmission of update $j+1$ is initiated only after update $j$ is delivered to all $k$ nodes in the group.}
% \node at [below = 0.25 of node_2] {\vdots};
\label{fig:sysmodel}
\negfigspace
\end{figure}
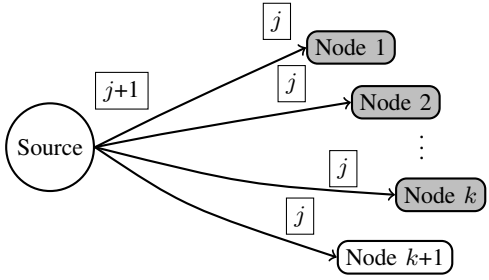

We consider a status updating system with a single source broadcasting time-stamped updates to multiple nodes through independent links with random delays, as shown in Fig.~\ref{fig:sysmodel}.
Each update message $j$ is time-stamped when it is generated at the source, and it takes time $X_{ij}$ to be successfully delivered to node $i$.
The priority group consists of  nodes $1,\ldots,k$, and the source guarantees the delivery of every update to all of these priority  nodes. 
In this work, we assume there is an instantaneous feedback channel from every node $i$ back to the source, and node $i$ acknowledges the source instantly as soon as the update is delivered to the node $i$. 
When all $k$ nodes in the priority group report receiving the update $j$, this update is considered completed and the transmissions of this update to all other nodes are terminated.
The source immediately generates the next update $j+1$ and repeats the multicast process. 

%We assume that the $X_i$ are i.i.d. shifted exponential $(\lambda,c)$ random variables. Consequently, each $X_i$ has CDF $F_{X}(x)~=~1-e^{-\lambda(x-c)}$, for $x\geq c$.
%The constant time shift $c>0$ captures the delay produced by the update generation and assembly process.
%On the other hand, it could also be a propagation delay on top of an exponential network delay if the source and database are geographically separated.

When most recently received update at time $t$ at node $i$ is time-stamped at time $u_i(t)$, the status update age or simply the {\it age}, is the random process $\age_i(t)=t-u_i(t)$. 
When an update reaches node $i$, $u_i(t)$ is advanced to the timestamp of the new update message. 
The time average of age process at a node, which is also called the age of information, is defined as 
\begin{align}
\age_i = \lim_{\tau\to\infty} \frac{1}{\tau} \int_{0}^{\tau} \age_i(t).
\end{align}

In this work, we derive the average age at both the priority node and the non-priority nodes, and we will show that the average age depends on the \emph{order statistics} of the random link delay $X_{ij}$.
\begin{definition} The $k$-th order statistic of random variables $X_1, \ldots, X_n$, denoted $X_{k:n}$, is  the $k$-th smallest variable. 
\end{definition}

For shifted exponential $X$ with CDF $F_X(x) = 1-e^{-\lambda(x-c)}$, $X_{k:n}$ has expectation and variance 
\begin{subequations}\label{os_both}
\begin{align}
	\eos{k}{n} = \E{X_{k:n}} & = \cf+\frac{1}{\lambda}(H_n-H_{n-k}), \label{os_first} \\
	\varos{k}{n} = \Var{X_{k:n}} & = \frac{1}{\lambda^2}\left(H_{n^2}-H_{(n-k)^2}\right),\label{os_var}
\end{align}
\end{subequations}
where $H_n$ and $H_{n^2}$ are the generalized harmonic numbers defined as $H_n = \sum_{j=1}^{n}\frac{1}{j}$ and $H_{n^2} = \sum_{j=1}^{n}\frac{1}{j^2}$.

\section{Priority Nodes} \label{sec:priority}

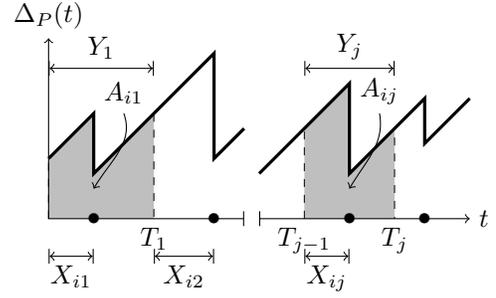
\begin{figure}[t]
%m-node
\centering
\begin{tikzpicture}[scale=0.2]
\draw [fill=lightgray, ultra thin, dashed] (0,0) to (0,4) to (3,7) to (3,3) to (7,7) to (7,0);
\draw [fill=lightgray, ultra thin, dashed] (17,0) to (17,6) to (20,9) to (20,3) to (23,6) to (23,0);
\draw [<-|] (0,12) node [above] {$\agep(t)$} -- (0,0) -- (13,0);
\draw [|->] (14,0) -- (28,0) node [right] {$t$};
\draw 
(3,0) node {$\bullet$}  
(7,0) node [below] {$T_1$}
%(15,0) node [below] {$Y_2$} 
(11,0) node {$\bullet$}
(17,0) node [below] {$T_{j-1}$} 
(20,0) node {$\bullet$} 
(23,0) node [below] {$T_{j}$}
(25,0) node {$\bullet$};
%\draw [|<->|] (25.3,0) to node [right] {$X_{n-1}$} (25.3,3); 
\draw  [<-] (3,2) to [out=60,in=290] (5,7) node [above] {$A_{i1}$};
%\draw[<-] (5,2) to [out=110,in=250] (5,7) node [above] {$Q_2$};
\draw[<-] (20,2) to [out=60,in=280] (22,7) node [above] {$A_{ij}$};
%sample path part 1
\draw [very thick] (0,4) -- (3,7) -- (3,3)  -- (7,7) 
-- (11,11)-- (11,4) --(13,6);
\draw [very thick] (14,3) -- (20,9) -- (20,3)  -- (25,8) -- (25,5) -- (28,8); 
\draw  [|<->|] (0,10) to node [above] {$Y_{1}$} (7,10);
\draw  [|<->|] (17,10) to node [above] {$Y_{j}$} (23,10);
\draw  [|<->|] (0,-2.5) to node [below] {$X_{i1}$} (3,-2.5);
%\draw  [|<->|] (3,-3) to node [below] {$X_2$} (7,-3);
\draw  [|<->|] (7,-2.5) to node [below] {$X_{i2}$} (11,-2.5);
\draw  [|<->|] (17,-2.5) to node [below] {$X_{ij}$} (20,-2.5);
\end{tikzpicture}
\caption{Sample path of the age $\agep(t)$ for node $i$ within the priority group with $k$ nodes. Update delivery instances are marked by $\bullet$.}
\label{fig:priority}\negfigspace
\end{figure}  

We start by evaluating the average age at a single node in the priority group.
Fig.~\ref{fig:priority} depicts a sample path of the age over time at some node $i$ in priority group with $k$ nodes.
Update $1$ begins transmission at time $t=0$ and is timestamped $T_0=0$. 
Here we remark that $X_{ij}$ is the \emph{service time} to deliver the update $j$ to node $i$.
Since the $X_{ij}$ are i.i.d. for all $i$ and $j$, the $\age_i(t)$ processes are statistically identical and each node $i$ has the same average age $\age_i$.
If one node gets an update earlier than any of the other $k-1$ nodes, it has to wait for an idle period until that update is delivered to all $k$ priority nodes.
The transmission time of an update $j$ to all $k$ nodes, which we call a \emph{service interval}, is given by
\begin{align}\eqnlabel{EY}
Y_j&=\max(X_{1j}, \ldots, X_{kj})=X_{k:k}.
\end{align}
%Note that $Y_j$ is identical to $X_{k:k}$.
We denote that update $j$ goes into service at time $T_{j-1}$ and gets delivered to all $k$ nodes at time $T_j = T_{j-1} + Y_j$.
%As shown in Figure \ref{fig:priority}, the age $\age_i(t)$ drops when update $j$ is delivered to node $i$ at an earlier time $T_{j-1}+X_{ij}$. 
%This implies at time $T_j$ when update $j$ completes transmission to all nodes, the age at monitor $i$ is $\age_i(T_i)=Y_j$.
Using similar techniques as in \cite{Zhong2017allerton}, we represent the area under the age sawtooth as the concatenation of the polygons $A_{i1},\ldots,A_{ij}$,  thus the average age is
\begin{align}\label{DeltaP}
\agep = \frac{\lim_{J\to\infty}\frac{1}{J} \sum_{j=1}^{J} A_{ij}}{\lim_{J\to\infty}\frac{1}{J}\sum_{j=1}^{J} Y_{ij}} = \frac{\E{A}}{\E{Y}}.
\end{align}
It follows from Fig.~\ref{fig:priority} that
\begin{align}
A_{ij}&=Y_{j-1}X_{ij} +X_{ij}^2/2 +X_{ij}(Y_j-X_{ij})+(Y_j-X_{ij})^2/2\nn
&=Y_{j-1}X_{ij}+Y_j^2/2.
\end{align}
Since $X_{ij}$ is independent of the transmission time $Y_{j-1}$ of the previous update, 
\begin{align}\label{EA}
\E{A}=\E{Y}\E{X}+\E{Y^2}/2.
\end{align} 
Denoting $\mu = \E{X}$, (\ref{os_both}), (\ref{DeltaP}) and (\ref{EA}) yield the next theorem.
\begin{theorem} 
\label{thm:age_p}
The average age at an individual node in the priority group is
\begin{align*}
\agep & = \mu + \frac{\eos{k}{k}}{2} + \frac{\varos{k}{k}}{2\eos{k}{k}}.
\end{align*}
\end{theorem}
Note that Theorem \ref{thm:age_p} is valid for any distribution of $X$. In terms of the Euler-Mascheroni constant $\gamma\approx0.577$, we also have the following result.

\begin{corollary}\label{cor:harmonics}
For shifted exponential $(\lambda,c)$ service time $X$, the average age at an individual node in the priority group is lower bounded by
\begin{align}
\agep \geq \frac{3c}{2} + \frac{1}{\lambda} + \frac{\log k+\gamma}{2\lambda}, \label{eqn:lowage_p}
\end{align}  \label{thm:lowage_p}
\end{corollary}
Proof appears in the Appendix. Corollary \ref{thm:lowage_p} indicates that the average age in the priority group $\agep$ is independent of the number of nodes $n$ in the system, and it behaves almost like a logarithmic function as the number of nodes $k$ increases. 

\section{Non-priority Nodes} \label{sec:nonp}

\begin{figure}[t]
\centering
\begin{tikzpicture}[scale=0.1]
%one slot is 10 units wide
%\draw[help lines] (0,0) grid (28,12);
\draw [<-|] (0,37) node [above] {$\agee(t)$} -- (0,0) -- (15,0);
\draw [|->] (18,0) -- (80,0) node [right] {$t$};
% \draw (0,0) \nxtic\nxtic\nxtic\nxtic \nxtic \nxtic;
\draw (0,0) (11,0) -- +(0,-1) (30,0) -- +(0,-1) (40,0) -- +(0,-1) (48,0) -- +(0,-1) (60,0) -- +(0,-1) (70,0) -- +(0,-1) ;
\fill[lightgray] (30,0) to ++(10,10) to ++(25,25) to ++(0,-30)  to ++(-5,-5);
\draw (7,0) node {$\bullet$}
(25,0) node {$\bullet$}
(38,0) node {$\bullet$}
(65,0) node {$\bullet$};
\draw (11,-6) node [above] {$T_1$} 
(30,-6) node [above] {$T_{j-1}$} (40,-6) node [above] {$T_{j}$} (48,-6) node [above] {$T_{j+1}$} (60,-6) node [above] {$T_{j+2}$} (70,-6) node [above] {$T_{j+3}$};
\draw [very thick] (0,10) -- ++(7,7) -- ++(0,-10) -- ++(8,8);
\draw [very thick]  (18,18) -- ++(7,7)  -- ++(0,-20) -- ++(7,7) -- (35,15) -- ++(3,3) -- ++(0,-10) -- ++(2,2) -- ++(25,25) -- ++(0,-30) -- ++(5,5);
% \draw [very thick, white] (15,15) -- ++(3,3);
\draw  [|<->|] (30,7) to node [below] {$Y_{j}$} (40,7);
\draw  [|<->|] (40,7) to node [below] {$Y_{j+1}$} (48,7);%36
\draw  [|<->|] (48,7) to node [below] {$Y_{j+2}$} (60,7);
\draw  [|<->|] (30,-7) to node [below] {$W_{l}$} (60,-7);
% \draw [|<-] (65,36) to (62,36) node [left] {$n-\Xtil$};
%\draw (60,45) node [left] {$n-\Xtil$};
% \draw [|<-] (70,36)  to (73,36) ;
\draw [|<-] (60,20) to (57,20);% node [left] {$\Xtil_n$};
\draw [|<-] (65,20) to (68,20) node [right] {$\Xtil$};
%\draw  [|<->|] (60,42) to node [below] {$\Xtil$} (65,42);
\draw  [|<->|] (72,0) to node [right] {$\Xtil$} (72,5);
% \draw  [|<->|] (72,5) to node [right] {$n-\Xtil_n$} (72,10);
% \draw [thin,dashed] (40,0) -- (40,10) -- (65,10);
% \draw [thin,dashed] (65,0) --(65,5) -- ( 70,5);
\draw[<-] (55,20) to [out=60,in=280] (55,30) node [above] {$A_{l}$};
\end{tikzpicture}
\caption{\small Sample path of the age $\agee(t)$ in the non-priority group: successful update deliveries (at times marked by $\bullet$) occur in intervals $1$, $j-1$, $j$, and $j+3$. Updates to the node are terminated in intervals $j+1$ and $j+2$.}
\label{fig:nonp}
\negfigspace
\end{figure}
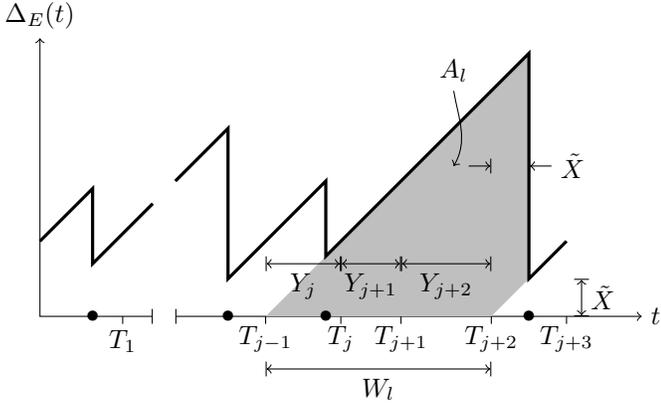

For a node in the non-priority group, the transmission of the current update is terminated right after the delivery of the update to all the $k$ nodes in the priority group.  
That is, a non-priority node $i$ fails to receive the update $j$ if and only if the service time $X_{ij}$ is larger than the service times of all the $k$ nodes in the priority group. 
Let's further denote that the priority nodes have service times $X_1, X_2, \ldots, X_k$, and the non-priority node has service time $X_{k+1}$.
For i.i.d. service time $X$, the probability that $X_{k+1}$ is the largest among all $k+1$ nodes is simply $q = {1}/(k+1)$,
since the rank of $X_{k+1}$ among $k+1$ random variables is uniform from $1$ to $k+1$. 
Here we also refer to $q$ as the failure probability for a non-priority node.

If update $j$ is not delivered to a node $i$, the node waits for a service interval $Y_j$ until the source generates the next update. 
Suppose an update is delivered to node $i$ during service interval $j$ and the next successful update delivery to node $i$ is in service interval $j+M_{l}$. 
In this case, $M_{l}$ is a geometric r.v. with probability mass function (PMF) 
$P_M(m) = q^{m-1} (1-q)$, and first and second moments
\begin{align}
\E{M} & = \frac{1}{1-q} = \frac{k+1}{k},
\label{eqn:EM} \\
\E{M^2} & = \frac{1+q}{(1-q)^2} = \frac{(k+1)(k+2)}{k^2}. \label{eqn:EM2}
\end{align}
We remark that $M_{l}$ and $Y_j$ are independent. 

An example of the age process is shown in Figure ~\ref{fig:nonp}. The update $j$ is delivered in service interval $j$ with end time $T_j$, and the node $k+1$ waits for $M_{l}=3$ service intervals until the next successful delivery in interval $j+3$. 
We represent the shaded trapezoid area as $A_{l}$ and the length in time between two service intervals with successful updates as $W_{l}=\sum_{j'=j}^{j+M_{l}-1} Y_{j'}$.
The time-averaged age for non-priority node is then
\begin{align}
\agee = \frac{\lim_{L\to\infty}\frac{1}{L} \sum_{l=1}^{L} A_{l}}{\lim_{L\to\infty}\frac{1}{L}\sum_{l=1}^{L} W_{l}} = \frac{\E{A}}{\E{W}}. \label{eqn:age_nonp_raw}
\end{align}
Denote the random variable $\Xtil_{j}$ as the service time of a successful update sent to a non-priority node with CDF $F_{\Xtil_{j}}(x) = F_{X_{j}|X_{j}<Y_j}(x)$. Evaluating Fig.~\ref{fig:nonp}, we have 
\begin{align}
	A_{l} & = \frac{1}{2} \left(W_{l} + \Xtil_{j}\right)^2 - \frac{1}{2}\Xtil_{j}^2
    %\nn& 
    = \frac{W_{l}^2}{2} + \Xtil_{j} W_{l}.\label{eqn:Ak}
\end{align}
Since $\Xtil$ and $W$ are independent, the average age in \eqref{eqn:age_nonp_raw} is
\begin{align}
	\agee & = \frac{1}{\E{W}} \Bigr( \frac{1}{2} \E{W^2} + \E{\Xtil}\E{W} \Bigr) \nn
	& = \frac{\E{W^2}}{2\E{W}} + \E{\Xtil} .  \label{eqn:EA}
\end{align}

We first define $Y_S$ as the length of a service interval given that the update is successfully delivered to the non-priority node $k+1$, and $Y_F$ as the length of a service interval given that the update is failed to be delivered. 
Thus, $Y_S$ and $Y_F$ have PDFs $f_{Y_S}(y) = f_{Y|Y\geq X_{k+1}}(y)$ and $f_{Y_F}(y) = f_{Y|Y<X_{k+1}}(y)$, respectively.
Furthermore, $\E{Y} = q\E{Y_F} + (1-q)\E{Y_S}$.

\begin{lemma}\label{thm:EW} $W$ has first and second moments
    \begin{align}
		\E{W} &= \E{M}\E{Y},\label{eqn:EW}\\
		\E{W^2} &= (\E{M}-1) \Var{Y_F} + \Var{Y_S} \nn 
		&\qquad + \left( \E{M^2}-2\E{M}+1\right) (\E{Y_F})^2 \nn
		&\qquad + (\E{Y_S})^2 + 2(\E{M}-1)\E{Y_F}\E{Y_S}. \label{eqn:EW2}
        \end{align}
\end{lemma}
Proof of the lemma appears in the Appendix. Lemma~\ref{thm:EW} leads to the following result for non-priority nodes. 

\begin{theorem} \label{thm:age_nonp}
	The average age at an individual node in the non-priority group is 
	\begin{align*}
		\agee & = \frac{1}{k}\sum_{i=1}^{k}\eos{i}{k+1} + \delta_1(k) + \delta_2(k),
		% + \frac{\delta_k}{2(k+1)\eos{k}{k}},
	\end{align*}
	where we denote
	\begin{align*}
		% \delta_k & = \varos{k}{k+1} + k \varos{k+1}{k+1} \nn 
		% & \qquad + \frac{k+2}{k} \eos{k}{k+1}^2 + k \eos{k+1}{k+1}^2 + 2\eos{k+1}{k+1}\,\eos{k}{k+1}. 
		\delta_1(k) & = \frac{\varos{k}{k+1} + k \varos{k+1}{k+1}}{2(k+1)\eos{k}{k}}\nn
		%\delta_2(k) & = \frac{1}{2(k+1)\eos{k}{k}} \Bigr( \frac{k+2}{k} \eos{k}{k+1}^2 + k \eos{k+1}{k+1}^2 \nn & \qquad + 2\eos{k+1}{k+1}\,\eos{k}{k+1} \Bigr) . \\
        \delta_2(k) & = \frac{\frac{k+2}{k} \eos{k}{k+1}^2 + k \eos{k+1}{k+1}^2 + 2\eos{k+1}{k+1}\,\eos{k}{k+1}}{2(k+1)\eos{k}{k}}. 
	\end{align*}
\end{theorem}
\begin{proof}
	In \eqref{eqn:EA}, $\Xtil$ indicates the service time of a non-priority node $k+1$ given that $X_{k+1}<\max(X_1,\ldots,X_k)$. 
	This condition implies $X_{k+1}$ cannot be the largest among all $k+1$ nodes.
	Thus,
	\begin{align}
		\E{\Xtil} & = \E{X_{k+1}\, |\,X_{k+1}<X_{k+1:k+1}} 
       %\nn&
       = \frac{1}{k} \sum_{i=1}^{k} \eos{i}{k+1} . \label{eqn:xtil}
	\end{align}
	The claim follows by substituting \eqref{eqn:EW}, \eqref{eqn:EW2} and \eqref{eqn:xtil} back into \eqref{eqn:EA}, and replacing  $\E{M}$ and $\E{M^2}$ by \eqref{eqn:EM} and \eqref{eqn:EM2}.	
\end{proof}
For exponential service times, Theorems~\ref{thm:age_p} and \ref{thm:age_nonp} yield the next claim.
\begin{theorem} \label{thm:expsame}
	For exponential service time $X$, the average age is the same for both priority and non-priority nodes and is given by 
	\begin{align}
		\agee = \agep = \frac{1}{\lambda} + \frac{H_k}{2\lambda} + \frac{H_{k^2}}{2\lambda H_k}, \label{eqn:expsame}
	\end{align}
	where $k$ is the priority group size. 
\end{theorem}

Theorem~\ref{thm:expsame} implies that the average age is identical for both groups regardless of whether an update is delivered to a node or not.

\section{Evaluation} \label{sec:evaluation}

\begin{figure}[t]
\centering

\begin{subfigure}[b]{0.5\textwidth}
\centering
\includegraphics[width=\textwidth]{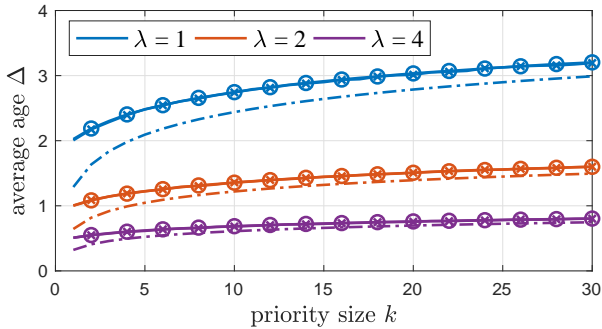}
\caption{exponential $X$. }
\label{fig:exp}
\end{subfigure}

\begin{subfigure}[b]{0.5\textwidth}
\centering
\includegraphics[width=\textwidth]{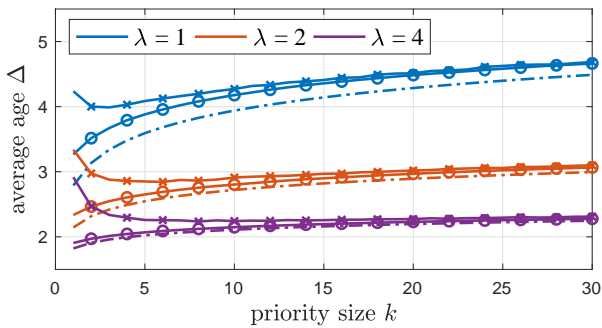}
\caption{shifted exponential $X$ with $\cf=1$. }
\label{fig:shiftexp}
\end{subfigure}
\caption{Average age versus the priority group size $k$. circle $\circ$ marks the priority group and cross $\times$ marks the non-priority group. The lower bound for priority group is shown as dashed line.}
\end{figure}

%\begin{figure}[t]
%\centering
%\begin{subfigure}[b]{0.5\textwidth}
%	\centering
%	\includegraphics[width=\textwidth]{shiftlamb1}
%	\caption{$\lambda=1$}
%	\label{fig:shiftlamb1}
%\end{subfigure}
%\begin{subfigure}[b]{0.5\textwidth}
%	\centering
%	\includegraphics[width=\textwidth]{shiftlamb2}
%	\caption{$\lambda=2$}
%	\label{fig:shiftlamb2}
%\end{subfigure}
%\caption{Average age versus the exponential shift $\cf$ with different rate $\lambda$. The priority group size $k=5$.}
%\end{figure}

Figures~\ref{fig:exp} and~\ref{fig:shiftexp} compare the simulation results of the average age for the priority group $\agep$ and the non-priority group $\agee$ as a function of the priority group size $k$. 
In Fig. \ref{fig:exp}, the link delay to every node $i$ is exponentially distributed with different $\lambda$. 
The average age curves for both groups overlap with each other and increases monotonically, which matches Theorem
 \ref{thm:expsame}. 
The lower bound on the average age for the priority group in Corollary \ref{thm:lowage_p} captures the trend for varying $k$, and becomes tighter for sufficiently large $k$. 
Fig. \ref{fig:shiftexp} shows the similar result for shifted exponential delay with $\cf=1$. 
For small group size $k$, there is a significant difference between the average age for two groups.
As $k$ increases, the age for non-priority group $\agee$ decreases slightly in the beginning and climbs up after a certain $k$. 
We also observe that the age difference between two groups vanishes for large enough $k$. 

\begin{figure}[t]
	\centering
	\includegraphics[width=0.5\textwidth]{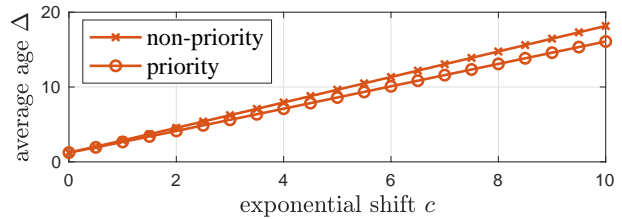}
	\caption{Average age versus the exponential shift $\cf$ with $\lambda=2$. The priority group size $k=5$.}
	\label{fig:shift}
\end{figure}

Fig. \ref{fig:shift} depicts the average age as a function of the shift parameter $\cf$ for shifted exponential delay $X$.
In Fig. \ref{fig:shift} with exponential rate $\lambda=2$, both groups have almost linear increasing average age for different the constant shift $\cf$.
The two curves start at the same point for $\cf=0$, and the difference in slopes leads to a larger gap between two curves as $\cf$ increases.

\section{Conclusion} \label{sec:conclusion}

In this work, we examine a status updating multicast network where the receivers are prioritized in terms of packet deliveries.
The average age at each receiver depends on the order statistics of the random link service time.
If the service time is identically distributed as exponential for every link, we show analytically and numerically that the average age at a priority node with packet delivery guarantee is the same as that without the guarantee.
The difference between two types of nodes arises if the exponential service time is mixed with a non-zero constant time shift. 
The analysis in this work is limited to exponential class service time, but we believe the difference between two types of nodes is related to the hazard rate of the service distribution, which potentially determines when the source should preempt the service of the current update with a fresh update.
\appendix 
\begin{proof}{\em Corollary~\ref{cor:harmonics}}
	Substituting (\ref{os_first}) and (\ref{os_var}) into Theorem~\ref{thm:age_p} gives 
	\begin{align}
		\agep & = \frac{3c}{2} + \frac{1}{\lambda} + \frac{H_k}{2\lambda} + \frac{H_{k^2}}{2\lambda^2\cf+2\lambda H_k}. \label{eqn:expage_p}
	\end{align}
	Note that $H_{k^2} = \sum_{i=1}^{k} \frac{1}{k^2}$ is monotonically increasing for $n\in\mathbb{Z^+}$ and $\lim_{k\to\infty} H_{k^2}=\pi^2/6$. Thus, given $\lambda$ and $\cf$,
	\begin{align}
		\lim_{k\to\infty}~\frac{H_{k^2}}{2\lambda^2\cf+2\lambda H_k}~\downarrow ~0. \label{eqn:low1}
	\end{align}
The harmonic number is given asymptotically by $H_k\approx\log k+\gamma+O(\frac{1}{k})$, which can be lower bounded by
\begin{align}
	H_k\geq\log k+\gamma, \qquad \textrm{for } k \in \mathbb{Z}_{>0}. \label{eqn:low2}
\end{align}	
Thus, \eqref{eqn:lowage_p} is given by substituting \eqref{eqn:low1} and \eqref{eqn:low2} into \eqref{eqn:expage_p}.
\end{proof}
\begin{proof}{\em Lemma~\ref{thm:EW}}
	We note the sequence $Y_j, \ldots, Y_{j+M_l-1}$ and the number of summation terms $M_l$ are dependent. 
	Since $M_l$ is geometric, the event $M_l=m$ indicates a sequence of $m-1$ consecutive failures followed by a success. Thus, $Y_{j'}$ is identical to $Y_F$ for $j'=\in\set{j,\ldots,j+M_l-2}$ and the last variable in the sequence $Y_{j+M_l-1}$ is identical to $Y_S$. This implies
	\begin{align}
		\E{W} & = \sum_{m=1}^{\infty} P_M(m) \E{\sum_{i=1}^{m}Y_i \Big| M=m} \nn
		& = \sum_{m=1}^{\infty} P_M(m) \Bigl( (m-1)\E{Y_F}+\E{Y_S} \Bigl) \nn
		%& = \E{Y_F} \sum_{m=1}^{\infty} P_M(m) (m-1) + \E{Y_S}  \sum_{m=1}^{\infty} P_M(m) \nn
		& = \E{Y_F} (\E{M}-1) + \E{Y_S}. \label{eqn:EW_YFYS}
	\end{align}
	Substituting \eqref{eqn:EM} into \eqref{eqn:EW_YFYS} yields
	\begin{align}
		\E{W} %& =\frac{1}{k} \E{Y_F} + \E{Y_S} \nn
		&= \frac{k+1}{k} \left( \frac{1}{k+1} \E{Y_F} + \frac{k}{k+1} \E{Y_S} \right)\nn
		% & = \frac{k+1}{k} \Bigr( q \E{Y_F} + (1-q) \E{Y_S} \Bigr) \nn
		&= \E{M} \E{Y}. 
	\end{align}	
For the second moment, we write $\E{W^2}$ in total expectation as
	\begin{align}
		\E{W^2} & = \sum_{m=1}^{\infty} P_M(m) \E{\Big(\sum_{i=1}^{m}Y_i \Big)^2 \Big| M=m} \nn
		& = \sum_{m=1}^{\infty} P_M(m) \left( \Var{\sum_{i=1}^{m}Y_i} + \left( \E{\sum_{i=1}^{m}Y_i } \right)^2 \right). \label{eqn:omega}
	\end{align}
	Since the random variables $Y_i$ are independent, we let
	\begin{align}
		\omega_1 & = \sum_{m=1}^{\infty} P_M(m) \Var{\sum_{i=1}^{m}Y_i} \nn
		& = \sum_{m=1}^{\infty} P_M(m) \Bigr( (m-1) \Var{Y_F} + \Var{Y_S} \Bigr) \nn
		& = \Var{Y_F} \Big(\E{M}-1\Big) + \Var{Y_S}. \label{eqn:omega1}
	\end{align}
	Similarly, we have  
	\begin{align}
		\omega_2 & = \sum_{m=1}^{\infty} P_M(m) \left( \E{\sum_{i=1}^{m}Y_i } \right)^2 \nn
		& = \sum_{m=1}^{\infty} P_M(m) \Bigr((m-1)\E{Y_F} + \E{Y_S}\Bigr)^2 \nn
		%& = (\E{Y_F})^2 \sum_{m=1}^{\infty} P_M(m) (m-1)^2  + (\E{Y_S})^2 \sum_{m=1}^{\infty} P_M(m) \nn
		%& \qquad +  \E{Y_F}\E{Y_S} \sum_{m=1}^{\infty} P_M(m) \; 2(m-1) \nn
		& = \left( \E{M^2}-2\E{M}+1\right) (\E{Y_F})^2 \nn
		&\qquad  + 2(\E{M}-1)\E{Y_F}\E{Y_S} + (\E{Y_S})^2. \label{eqn:omega2}
	\end{align}
The claim follows by substituting \eqref{eqn:omega1} and \eqref{eqn:omega2} in \eqref{eqn:omega}. 
\end{proof}

\begin{proof}{\em Theorem~\ref{thm:expsame}}
	For priority nodes, we obtain the average age by substituting \eqref{os_first} and \eqref{os_var} to Theorem \ref{thm:age_p} with $c=0$, which directly yields \eqref{eqn:expsame}. For non-priority nodes, the first term in Theorem \ref{thm:age_nonp} is 
	\begin{subequations}
	\begin{align}
		\delta_0(k) & = \frac{1}{k}\sum_{i=1}^{k}\eos{i}{k+1} \\
		& = \frac{1}{k}\sum_{i=1}^{k} \frac{H_{k+1}-H_{k+1-i}}{\lambda} \\
		& = \frac{H_{k+1}}{\lambda} - \frac{1}{\lambda k} \sum_{i=1}^{k} H_i \\
		& = \frac{H_{k+1}}{\lambda} - \frac{k+1}{\lambda k} (H_{k+1}-1) \label{eqn:sumharmonic} \\
		& = \frac{1}{\lambda} + \frac{1}{\lambda k} - \frac{H_{k+1}}{\lambda k}. \label{eqn:delta_0}
	\end{align}\end{subequations}	
	In \eqref{eqn:sumharmonic}, we use the series identity of Harmonic numbers $\sum_{i=1}^{k} H_i = (k+1) (H_{k+1}-1)$.
	Similarly, substituting \eqref{os_first} and \eqref{os_var} into $\delta_1(k)$ and $\delta_2(k)$ gives
	\begin{align}
		\delta_1 (k)
		% & = \frac{\varos{k}{k+1} + k \varos{k+1}{k+1}}{2(k+1)\eos{k}{k}}  \nn 
		& = \frac{(H_{(k+1)^2}-1) + kH_{(k+1)^2}}{2(k+1)\lambda H_k} \nn
		& = \frac{H_{k^2}}{2\lambda H_k} - \frac{k}{2\lambda(k+1)^2 H_k}, \label{eqn:delta_1} \\
		\delta_2 (k) 
		% & = \frac{k+2}{2k(k+1)} \frac{\eos{k}{k+1}^2}{\eos{k}{k}}  + \frac{k}{2(k+1)} \frac{\eos{k+1}{k+1}^2}{\eos{k}{k+1}^2} \nn
		% & \qquad + \frac{1}{k+1} \frac{\eos{k+1}{k+1}\,\eos{k}{k+1}}{\eos{k}{k+1}^2} \nn
		& = \frac{k+2}{2k(k+1)} \frac{(H_{k+1}-1)^2}{\lambda H_k}  + \frac{k}{2(k+1)} \frac{H_{k+1}^2}{\lambda H_k} \nn
		& \qquad + \frac{1}{k+1} \frac{(H_{k+1}-1)H_{k+1}}{\lambda H_k} \label{eqn:delta_2} .
	\end{align}
	We note that $\delta_0(k)$ in \eqref{eqn:delta_0} and $\delta_2(k)$ in \eqref{eqn:delta_2} only contain first order harmonic numbers, thus we combine two terms and rewrite $H_{k+1} = H_k + 1/(1+k)$, which gives
	\begin{align}
		\delta_0(k)  + \delta_2 (k)  = \frac{1}{\lambda} + \frac{H_k}{2\lambda} + \frac{k}{2\lambda(k+1)^2 H_k} \label{eqn:delta_sum}
	\end{align}
	The claim is given by the sum of \eqref{eqn:delta_1} and \eqref{eqn:delta_sum}.
\end{proof}

%\end{proof}%\negspacesmall

\section*{Acknowledgment}
\addcontentsline{toc}{section}{Acknowledgment}
Part of this research is based upon work supported by the National Science Foundation under grants CNS-1422988 and and CCF-1717041.

\bibliographystyle{IEEEtran}
\bibliography{ref}

\end{document}